\documentclass[sn-mathphys-num]{sn-jnl}
\usepackage{multirow}%
\usepackage{amsmath,amssymb,amsfonts}%
\usepackage{amsthm}%
\usepackage{mathrsfs}%
\usepackage[title]{appendix}%
\usepackage{xcolor}%
\usepackage{textcomp}%
\usepackage{manyfoot}%
\usepackage{booktabs}%
\usepackage{algorithm}%
\usepackage{algorithmicx}%
\usepackage{algpseudocode}%
\usepackage{listings}%
\newtheorem{theorem}{Theorem}
\newtheorem{lemma}{Lemma}
\newtheorem{problem}{Problem}
\newtheorem{assumption}{Assumption}
\newtheorem{remark}{Remark}%

\newtheorem{definition}{Definition}%

\begin{document}

\title[Article Title]{Distributed Cooperative Control with Prescribed Performance of BESSs with A Unified Discharge Constrain for Power Allocation under Dynamic Load}

\author[1,2]{\fnm{Yalin} \sur{Zhang}}\email{zhangyl@mail.nankai.edu.cn}

\author[1,2]{\fnm{Zhongxin} \sur{Liu}}\email{lzhx@nankai.edu.cn}

\author*[1,2]{\fnm{Fuyong} \sur{Wang}}\email{wangfy@nankai.edu.cn}
\author[1,2]{\fnm{Zengqiang} \sur{Chen}}\email{chenzq@nankai.edu.cn}

\affil*[1]{\orgdiv{College of Artificial Intelligence}, \orgname{Nankai University}, \orgaddress{\street{No. 38 Tongyan Road}, \city{Tianjin}, \postcode{300350}, \state{Tianjin}, \country{China}}}

\affil*[2]{\orgdiv{Tianjin Key Laboratory of Interventional Brain-Computer Interface and Intelligent Rehabilitation}, \orgname{Nankai University}, \orgaddress{\street{No. 38 Tongyan Road}, \city{Tianjin}, \postcode{300350}, \state{Tianjin}, \country{China}}}
\abstract{Battery energy storage system (BESS) is integrated into the smart grid to enhance scalability, economy, and greenery. And the State-of-Charge (SoC) balance is one of the basic problems of BESSs, which can maximize the utilization of capacity. BESSs with a unified relative variation rate for SoC can be simultaneously filled or empty, while real-time estimation schemes of SoC balance and power sharing states are required in this power allocation scheme. Therefore, the prescribed performance control (PPC) method is applied in this paper to design two distributed estimators based on multi-agent systems (MASs), in order to estimate the power sharing and SoC balance states in real-time under dynamic load driving. In this way, these two average values can ultimately be well estimated with almost zero error performance, and dynamic performance and steady-state performance of the two estimators can be adjusted by different parameters. Similarly, consensus performance and dynamic tracking performance are decoupled. These results provide a broader range for the selection of gains. To verify the effectiveness, robustness and progressiveness of the designed estimators, some cases with a resistance network containing 4 BESSs as load distribution are designed and discussed. Further more, to test scalability, a large-scale system containing 12 BESSs is conducted to the designed scheme.}

\keywords{Multi-agent systems, distributed dynamic average tracking control, prescribed performance, power allocation, battery energy storage system}

\maketitle
\section{Introduction}
\par The traditional large thermal power units are gradually being replaced by distributed renewable energy generation, which greatly reduces power generation costs and protects the environment. At the same time, battery energy storage system (BESS) has become an indispensable part of the smart grid to ensure the provision of high-quality uninterrupted power \cite{HossainLipu2022, WANG2022104812, solyaliComprehensiveStateoftheartReview2022, ghanjatiOptimalSizingEnergy2022}. This is mainly because renewable energy generation, as a distributed power source, exhibits intermittency and randomness \cite{WANG2022104812, solyaliComprehensiveStateoftheartReview2022, ghanjatiOptimalSizingEnergy2022}. In addition, BESS can alleviate the power gap between supply and demand, thereby alleviating the electricity shortage in the community \cite{HossainLipu2022, WANG2022104812}. Therefore, BESS, as a backup power source or supplement, is extensively integrated into the smart grid, as shown in Fig. \ref{EI}, to improve its stability, economy, and greenness \cite{HossainLipu2022, WANG2022104812, solyaliComprehensiveStateoftheartReview2022}.
\begin{figure}
	\centering
	\includegraphics[width=10cm]{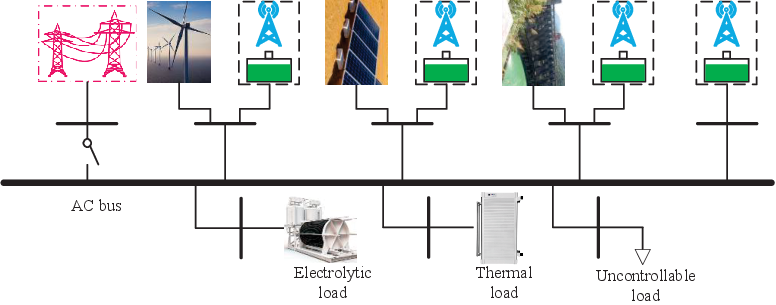}
	\caption{BESSs in smart grids.}
	\label{EI}
\end{figure}
\par However, battery cells in parallel BESSs may waste capacity due to insufficient charging and discharging, which cannot maximize the economic and green efficiency of the power grid. This is due to the State of Charge (SoC) \cite{Zeng2022l, Wang2020d, Litjens2018} imbalance caused by inappropriate power allocation schemes. Therefore, the energy management and power allocation issues of BESS are highly worth studying to achieve SoC balance. 
\par At first, in order to achieve SoC balance, scholars designed some centralized or decentralized solutions. Among them, decentralized schemes represented by discharge circuits \cite{Baveja2023} were designed to consume excess power to ensure that other batteries match the one at the lowest SoC level. This solution resulted in waste of capacity and electricity. And some centralized schemes, such as SoC compensators \cite{Cao2021}, were designed to effectively achieve SoC balancing. However, centralized methods are often replaced by distributed solutions that have emerged in recent years due to the need for expensive control centers, susceptibility to single point failures, and poor scalability.
\par As a typical distributed system, multi-agent systems (MASs) exhibit high efficiency, high scalability, and robustness, and are widely used to manage microgrids/BESSs. Under constant
load, some static distributed controllers \cite{Cai2016, Khazaei2019a, Ding2020, Hu2019, Nguyen2021a, Zhang2020} were designed to achieve SoC balance. Therein, some distributed controllers applied for the first-order model based on the ampere hour integration method were designed to directly balance the SoC \cite{Khazaei2019a, Zhang2020}. Others were designed to be suitable for the second-order integrator model composed of SoC and active power \cite{Cai2016, Ding2020, Hu2019, Nguyen2021a}, in order to achieve SoC balance and power sharing. Based on these two modes, some researchers have developed finite time/fixed time protocols \cite{Ding2020, Hu2019}, fault tolerance/attack resilience protocols \cite{Ding2020}, etc., in order to improve the dynamic and steady-state performance as well as robustness of the system. In addition, methods such as model free control \cite{Hong2021d} and sliding mode control \cite{Zhang2019e} were used to improve the static protocol for SoC balance. The application of the above solutions can effectively solve the SoC balance problem under constant load and demonstrate different good performance. However, theoretically speaking, these protocols seem to be ineffective in addressing SoC balance issues under time-varying loads.
\par Recently, some preliminary explorations have been completed in \cite{Xing2019, Meng2021, Meng2022, Wu2022} to address the SoC balance problem under time-varying loads. BESSs with a unified SoC relative variation rate often achieved simultaneous filling or emptying \cite{Xing2019}. This power allocation scheme required real-time estimation of SoC balance and power sharing states, which provided us with vast exploration space and endless inspiration. Therefore, with the help of the distributed control method for MASs, the states of SoC balance and power sharing could be estimated in a distributed manner. An asymptotic distributed estimation scheme was first designed in \cite{Xing2019}. Subsequently, an asymptotic scheme with an adjustable steady-state error and a finite time scheme were designed simultaneously in \cite{Meng2021}. Furthermore, for the asymptotic scheme with adjustable steady-state error, the parameters of the battery were estimated in \cite{Meng2022}. A discrete scheme was designed in \cite{Wu2022} to achieve consensus in multiple relative variation rates, thereby achieving SoC balance and power sharing. In \cite{9925608}, a distributed prescribed-time consensus scheme was constructed to accelerate convergence. On the basis of \cite{Xing2019}, some authors have designed a distributed scheme in \cite{luDistributedSecureBalancing2023} that can resist DoS attacks. Similarly, a distributed scheme to resist random communication failures was developed in \cite{DistributedStateofchargePower2024}. Furthermore, in order to save communication resources, two distributed schemes with event triggered communication mechanisms were developed in \cite{qianDistributedEventtriggeredAlgorithms} and \cite{xingRobustEventTriggeredDynamic2020b}, respectively.
\par It is obvious that the asymptotic estimator mentioned above can only adjust steady-state errors \cite{Xing2019, Meng2021, Meng2022}, while finite- \cite{9925608}/prescribed-time estimators \cite{Meng2021} often inevitably exhibit chatting. Therefore, a high-quality estimation scheme that balances dynamic and steady-state performance is urgently needed for BESSs with a unified relative variation rate.
\par Inspired by the above work, for BESSs with unified SoC relative variation rate, in order to design decoupled dynamic and steady-state performance schemes, the concept of prescribed performance control (PPC) \cite{Bechlioulis2008, Bechlioulis2009, Bechlioulis2010} is introduced to improve the asymptotic estimators. Specifically, the main contributions are presented as follows:
\begin{enumerate}	
	\item  In the designed scheme, information about the derivative of the load with respect to time is not required. This is different from \cite{Meng2021, 9925608}.
	\item The consensus performance and average tracking performance are decoupled due to the application of error modulation technology. In this way, when selecting gains, there is no trade-off between various performances, providing a wider range of choices.
	\item Compared to the prescribed-time estimators designed in \cite{9925608}, the scheme adopted exhibits superiority in terms of consensus, average tracking, and steady-state performance.
\end{enumerate}
\par Provide an introduction for the following chapters. Firstly, relevant issues, including the dynamic of SoC, an unified relative SoC change rate, and control objectives, are described in Section \ref{2}. Secondly, some preliminary knowledge, such as algebraic graph theory, PPC method, and dynamical system, is introduced in Section \ref{3}. Subsequently, the constraints of load and SoC, and power sharing and SoC balance states estimators are designed in Section \ref{4}. Some cases are designed to validate and test the adopted estimators in Section \ref{5}. In the end, a conclusion is drawn in Section \ref{6}.
\section{Problem Statement}\label{2}
In this section, the dynamic of the SoC of a battery based on the ampere hour integration method is modified to the active power integration scheme. A discharging rate constraint is introduced to manage batteries to ensure that no battery will exit prematurely due to insufficient power. Based on this, the control objective of this article is given. For simplicity, unless necessary, the time variables in all symbols are omitted by default.
\subsection{SoC and A Discharge Rate Constraint of A Battery}
For a BESS, the energy storage level of the battery unit is characterized by SoC, and the ampere hour integration method \cite{Cao2021, Cai2016, Hong2021d, Xing2019, Chen2023b, Ren2020}, as shown in \eqref{soc}, under different time scales,
\begin{equation}
	\label{soc}
	\dot{E}_i=-\frac1{\mathrm{T}\mathrm{Q}_i}I_i,
\end{equation}
is an effective way for calculating SoC by relevant experimental analysis,
where $E_i$ is the SoC of BESS $i$, $\mathrm{Q}_i$ is a constant parameter related to the capacity of BESS $i$, $\mathrm{T}$ denotes the time scale. Besides, the output power of BESS $i$ is calculated as $\tilde{P}_{i} = \mathrm{V}_{i}I_{i}$, where $\mathrm{V}_i$ is the terminal voltage of BESS $i$. Considering the small change in battery terminal voltage when the SoC is in a large range, define a constant $\mathrm{K}^E_i$ as $\mathrm{K}_{i}^{E} = \frac{1}{\mathrm{Q}_{i}\mathrm{V}_{i}}$, which is the charge/discharge coefficient and implies that SoC conforming to the dynamic \eqref{soc} is heterogeneous. And then, the SoC of BESS $i$ for $i\in\{1,2,\cdots,\mathrm{n}\}$ is subjected to
\begin{equation}
	\label{soc1}
	\dot{E}_i=-\frac{\mathrm{K}_i^E}{\mathrm{T}}\tilde{P}_i=-\frac{1}{\mathrm{T}}P_i,
\end{equation}
where $P_i=\mathrm{K}_i^E\tilde{P}_i$ is proportional output power of BESS $i$ \cite{Cai2016, Khazaei2019a, Ding2020, Hu2019, Nguyen2021a, Zhang2020, Zhang2019e, Xing2019, Meng2021, Meng2022, Wu2022}. For simplicity, $P_i$ for $i\in\{1,2,\cdots,\mathrm{n}\}$ is referred to as power in this paper.
\par SoC balance is one of the key problems in managing parallel BESSs, which ensures that no battery will exit prematurely due to low energy level. In a multi-BESS network, achievement of SoC balance can greatly optimize battery life. Supply and demand balance should be maintained. However, the power of each BESS is often constrained by SoC and is expected to have the same relative SoC change rate \cite{Xing2019, Meng2021, Meng2022, Wu2022}, i.e.,
$$\frac{P_i}{E_i}=\frac{P_\mathrm{a}}{E_\mathrm{a}},$$
where $E_\mathrm{a}$ and $P_\mathrm{a}$ are the values of the average SoC and the average power of BESS $i$, and which is rewritten as
\begin{equation}
	\label{power0}
	P_i=\frac{P_\mathrm{a}}{E_\mathrm{a}}E_i.
\end{equation}
Here, a theorem is cited to illustrate the impact of a power generation rule as described in \eqref{power0} on a multi-BESS network.
\begin{theorem}
	\cite{Xing2019, Meng2021, Meng2022, Wu2022} For a multi-BESS network containing $\mathrm{n}$ BESSs, the simultaneous depletion of all BESSs can be ensured by \eqref{power0}. Furthermore, the balance between power supply and demand can be reached if $P_\mathrm{a}$ is the average load power.
\end{theorem}
\par Thus, it is necessary to estimate the values of $P_\mathrm{a}$ and $E_\mathrm{a}$ to allocate power by
\begin{equation}
	\label{power1}
	P_i=\frac{\hat{P}_{\mathrm{a},i}}{\hat{E}_{\mathrm{a},i}}E_i,
\end{equation}
where $\hat{E}_{\mathrm{a},i}$ and $\hat{P}_{\mathrm{a},i}$ are the real-time estimated values of the average SoC and the average power by BESS $i$. For a BESS, it is usually determined whether it is in charging mode or discharging mode by defining the numerical sign of its output power. In this paper, a positive output power means it is in discharge mode, while a negative output power means it is in charge mode. The spatial control structures of grid-connected BESSs and isolated BESSs are different. Among them, droop control is suitable for isolated microgrids, while grid-connected microgrids are often controlled by power-quality (PQ) control. Both grid connected and isolated modes require a power reference value when implementing control. Thus, for isolated BESSs, this power allocation scheme is shown in Fig. \ref{Bcon}, where each BESS is in discharging mode and droop control is the primary control method.
\begin{figure}
	\centering
	\includegraphics[width=6cm]{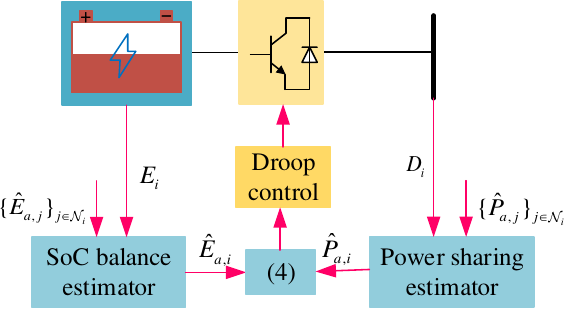}
	\caption{Control structure of a BESS under a discharge rate constrain.}
	\label{Bcon}
\end{figure}
\subsection{Control Objectives}
According to \eqref{power1}, in order to allocate power reasonably so that each BESS can unleash maximum capacity, the average SoC and the average power should be well estimated. The detailed objectives description are organized in Problem \ref{problem1}.
\begin{problem}
	\label{problem1}
	Consider a multi-BESS network including $\mathrm{n}$ BESSs. And the SoC of each BESS is obtained by dynamic \eqref{soc1}, the power can be calculated by \eqref{power1}. It is intended to design two distributed estimators for each BESS to solve the following problems,\\
	1) The real time power sharing state of each BESS can be well estimated, and the steady-state error can be limited to evolve within a predetermined zone, i.e.,
	$$\lim\limits_{t\to+\infty}\left|\hat{P}_{\mathrm{a},i}-P_\mathrm{a}\right|\leq\epsilon_P,$$
	where $i\in\{1,2,\cdots,\mathrm{n}\}$, $P_\mathrm{a}=\frac{1}{\mathrm{n}}\sum\limits_{i=1}^{\mathrm{n}}D_{i}$, $D_{i}$ is the load that BESS $i$ needs to share, $\epsilon_P$ is a pre-specified value.\\
	2) The real-time SoC balance state can be well estimated, and the steady-state error can be limited to evolve within a predetermined zone, i.e.,
	$$\lim\limits_{t\to+\infty}\left|\hat{E}_{\mathrm{a},i}-E_\mathrm{a}\right|\leq\epsilon_E,$$
	where $E_\mathrm{a}=\frac{1}{\mathrm{n}}\sum\limits_{i=1}^{\mathrm{n}}E_{i}$, $\epsilon_E$ is a pre-specified value.
\end{problem}
\section{Preliminaries}\label{3}
Algebraic graph theory is usually used to explain communication networks such that distributed controllers based on MASs are well designed. The PPC method is applied to regulate the steady-state and dynamic performance of the system. In this section, these two kinds of useful preliminary knowledge for controller design are introduced. In addition, lemmas related to the maximal solution problem of dynamical systems are introduced for stability analysis.
\subsection{Graph Theory}
For a multi-BESS network managed by a MAS, agents communicate with each other over a topology $\mathcal{G} = (\tilde{\mathcal{V}},\mathcal{E})$.
Therein, $\mathcal{V}=\{v_{i}|i\in\{1,2,\cdots,\mathrm{n}\}\}$ stands for the set embracing n nodes, each of which represents an agent. Meanwhile, $\mathcal{E}=\{(v_{i},v_{j})|v_{i},v_{j}\in\mathcal{V}\}$ is defined as the edge set. If there
exists an edge pointing from nodes $j$ to $i$, it is denoted as $(v_i,v_j)$, and agent $j$ is called as a neighbor of agent $i$, which means agent $i$ has access to the information of agent $j$. Denote $\mathcal{N}_{i}=\{v_{j}|(v_{i},v_{j})\in\mathcal{E}\}$ as the neighbor set of agent $i$.
\subsection{PPC Method}
The PPC method is implemented to ensure the dynamic and steady-state performance of the system. More specifically, PPC can enable scalar error $e(t)$ to converge to a smaller set of prescribed residuals with a minimum prescribed convergence rate \cite{Bechlioulis2008, Bechlioulis2009}. Once a prescribed performance controller is designed, in this way, the scalar error $e(t)$ strictly evolves within the predetermined region, i.e,
$$-\rho(t)<e(t)<\rho(t),\forall t\geq0,$$
where $\rho(t)$, called a performance function, is a function of time and smooth and bounded and required that $\lim\limits_{t\to+\infty}\rho(t)>0$.
\par The commonly used exponential performance function is shown as follows,
\begin{equation}
	\label{pf}
	\rho(t)=(\rho_0-\rho_\infty)e^{-\lambda t}+\rho_\infty,
\end{equation}
where $\rho_0$, $\rho_\infty$ and $\lambda$ are positive constants. Usually, $\rho_0$ is selected as $\rho_0>|e(0)|$, and $\rho_\infty$ is often considered a parameter that characterizes the resolution of measurement equipment and is the steady-state error allowed by the system. Besides, $\lambda$ is the allowable minimum convergence rate of $e(t)$. This
indicates that the introduction of performance function $\rho(t)$ constrains the steady-state error of the system and improves its dynamic performance.
\par Define a modulated error $\xi(t)=\frac{e(t)}{\rho(t)}$, and a strictly increasing, odd and bijective mapping $T{:}(-1,1)\to(-\infty,+\infty)$ with
$T(0) = 0$, shown in below, is employed,
\begin{equation}
	\label{pf1}
	T(\xi)=\frac{1}{2}\ln(\frac{1+\xi}{1-\xi}),
\end{equation}
whose Jacobian derivative is found as follows,
\begin{equation}
	\label{pf2}
	J_T(\xi)=\frac{\mathrm{d}T(\xi)}{\mathrm{d}\xi}=\frac{1}{1-\xi^2},
\end{equation}
such that $J_T(\xi)>0$.
\begin{lemma}
	\label{lem1}
	\cite{Bechlioulis2010} There is an important property on \eqref{pf1} and \eqref{pf2}. That is, there exists a positive constant $\bar{\epsilon}$ such that $|T(\xi)|>\bar{\epsilon}$, and for any $\mathrm{K}$, $\mathrm{F}>0$, $$-\mathrm{K}T(\xi)J_T(\xi)\xi+\mathrm{F}<0.$$
\end{lemma}
\subsection{Dynamical systems}
Consider the following initial value problem
\begin{equation}
	\label{Ds}
	\dot{x}=f(t,x), x(0)=x^0\in\Omega_x,
\end{equation}
where $f:\mathbb{R}_+\times \Omega_x \to \mathbb{R}^\mathrm{n}$ and $\Omega_x\subset\mathbb{R}^\mathrm{n}$ is an open and nonempty set. Next, it is necessary to give the following definition and conclusions regarding this problem.
\begin{definition}
	\label{def1}
	\cite{Sontag1998} A solution $x(t)$ of \eqref{Ds} is so-called maximal if it cannot be right extended.
\end{definition}
\begin{lemma}
	\label{lem2}
	\cite{Sontag1998} Consider the initial value problem \eqref{Ds} for $x\in\Omega_{x}$. If $f(t,x)$ is piecewise continuous and locally integrable on $t$ and locally Lipschitz on $x$, there exists a
	maximal solution $x(t)\in\Omega_{x}$ with $\forall t\ge 0$ of \eqref{Ds}.
\end{lemma}
\begin{lemma}
	\label{lem3}
	\cite{Sontag1998} Based on Lemma \ref{lem2}, for a finite time interval $t\in[0, t_\mathrm{max})$ with $t_\mathrm{max}<+\infty$, there must exist a compact subset $\Omega_{x}^{'}\subset\Omega_{x}$ and a time instant $t'$
	such that $x(t^{'})\notin\hat{\Omega_{x}^{'}}$.
\end{lemma}
\section{Design of Two Distributed Estimators with Prescribed Performance}\label{4}
In this section, two distributed estimators with prescribed performance are designed to estimate the power sharing and SoC balance states. Firstly, the characteristics of the load and BESS characteristics in the resistor network applicable to the adopted scheme are introduced. Subsequently, a continuous dynamic scheme with prescribed performance is designed, and stability is analyzed.
\subsection{Related Descriptions of A Resistive Network}
In this paper, the resistance network adopted here can be either AC or DC.
It should be noted that for AC microgrids, the allocation of active power is concerned in this paper. Therefore, Assumption \ref{assum2} is made for the active load here.
\begin{assumption}
	\label{assum2}
	Assuming that each bus is equipped with a BESS, then the local load $D_i$ is schedulable and also continuous, bounded, and has a first derivative.
\end{assumption}
\begin{remark}
	\label{rem1}
	Assuming that there exists total load $\tilde D$ in a resistive network, BESS $i$ for $i\in \{1, 2,\cdots, \mathrm{n}\}$ shares a certain amount of load $D_i$ in steady-state according to certain principles. Specially, a simple method to calculate the load that battery cell $i$ with $i\in \{1, 2,\cdots, \mathrm{n}\}$ should shared proportionally is as follows,
	$$D_i=\mathrm{K}_i^E\tilde{D}_i=\frac{1}{\sum_{j=1}^\mathrm{n}\frac{1}{\mathrm{K}_i^E}}\tilde{D}.$$
	Obviously, this method is a centralized calculation in steady-state, which is applied in \cite{Xing2019, Meng2021, Meng2022}. Hence, $D_i$ is continuous, bounded, and has a first derivative.
\end{remark}
\subsection{Distributed Method to Obtain The SoC Balance State And The Power Sharing State}
Inspired by \cite{stamouliRobustDynamicAverage2022}, the distributed SoC balance and load sharing estimators can be designed for each battery unit here, as shown in \eqref{powerea}-\eqref{powered},
\begin{subequations}
	\label{powere}
	\begin{equation}
		\label{powerea}
		\begin{aligned}
			\dot{z}_{i}^{P}=-\mathrm{k}_\mathrm{r}^Pz_{i}^{P}
			-\mathrm{k}^P\sum_{i\in\mathcal{N}_{i}}\rho_{ij}^{-1}(t)J_{T}(\frac{\hat{P}_{\mathrm{a},i}-\hat{P}_{\mathrm{a},j}}{\rho_{ij}})T(\frac{\hat{P}_{\mathrm{a},i}-\hat{P}_{\mathrm{a},j}}{\rho_{ij}(t)}),
		\end{aligned}
	\end{equation}
	\begin{equation}
		\label{powereb}
		\hat{P}_{\mathrm{a},i}=z_i^P+D_i,
	\end{equation}
	\begin{equation}
		\label{powerec}
		\begin{aligned}
			\dot{z}_i^E=-\mathrm{k}_\mathrm{r}^Ez_i^E
			-\mathrm{k}^E\sum_{j\in\mathcal{N}_i}\rho_{ij}^{-1}(t)J_T(\frac{\hat{E}_{\mathrm{a},i}-\hat{E}_{\mathrm{a},j}}{\rho_{ij}})T(\frac{\hat{E}_{\mathrm{a},i}-\hat{E}_{\mathrm{a},j}}{\rho_{ij}(t)}),
		\end{aligned}
	\end{equation}
	\begin{equation}
		\label{powered}
		\hat{E}_{\mathrm{a},i}=z_i^E+E_i,
	\end{equation}
\end{subequations}
where $\hat{E}_{\mathrm{a},i}$ and $\hat{P}_{\mathrm{a},i}$ are the estimation of $E_\mathrm{a}$ and $P_\mathrm{a}$, $z_i^P$ and $z_i^E$ are the designed intermediate states. Here, let initial condition ${\hat E}_{\mathrm{a},i}=E_i$, ${\hat P}_{\mathrm{a},i}=P_i$, $z_i^P=0$ and $z_i^E=0$ for $i\in\{1,2,\cdots,\mathrm{n}\}$.
\par In order to facilitate the interpretation of the proof, first define two consensus error vectors, i.e., $\delta^{E} = [\delta_{1}^{E},\cdots,\delta_{l}^{E},\cdots,\delta_\mathrm{m}^{E}]$ and $\delta^P = [\delta_{1}^P,\cdots,\delta_{l}^P,\cdots,\delta_\mathrm{m}^P]$, where $l$ and $\mathrm{m}$ are the index and the total number of edges of $\cal G$ respectively.
\par Next, define a diagonal matrix $R = \mathrm{diag}([\rho_{1},\cdots,\rho_\mathrm{m}])$. After normalizing vector $\delta^P$ and $\delta^E$, they can be obtained that
$$\xi^E=\mathrm{col}(\frac{\delta_1^E}{\rho_1(t)},\cdots,\frac{\delta_\mathrm{m}^E}{\rho_\mathrm{m}(t)})=R^{-1}(t)\delta^E,$$
$$\xi^P=\mathrm{col}(\frac{\delta_1^P}{\rho_1(t)},\cdots,\frac{\delta_\mathrm{m}^P(t)}{\rho_\mathrm{m}})=R^{-1}(t)\delta^P.$$
Applying the \eqref{pf1} and \eqref{pf2}, one can get
$$T^{E}=\mathrm{col}(T(\xi_{1}^{E}),\cdots,T(\xi_\mathrm{m}^{E})),$$
$$T^{P}=\mathrm{col}(T(\xi_{1}^{P}),\cdots,T(\xi_\mathrm{m}^{P})),$$
$$J_{T}^{E}=\mathrm{diag}([J_{T}(\xi_{1}^{E}),\cdots,J_{T}(\xi_\mathrm{m}^{E})]),$$
$$J_{T}^{P}=\mathrm{diag}([J_{T}(\xi_{1}^{P}),\cdots,J_{T}(\xi_\mathrm{m}^{P})]). $$
\par Therefore, a theorem and its proof are given as follows, which are on the stability and steady-state performance of the adopted power sharing state estimator scheme.
\begin{theorem}
	\label{TH1}
	Consider a resistive network containing $\mathrm{n}$ BESSs and satisfying Assumption \ref{assum2}. Under the estimator \eqref{powerea} and \eqref{powereb} and an undirected graph, the average power can be estimated by agent $i$ for all $i\in\{1,\cdots,\mathrm{n}\}$ with steady-state error as follows,
	\begin{equation}
		\lim_{t\to+\infty}|\hat{P}_{\mathrm{a},i}-D_\mathrm{a}|\leq\frac{\mathrm{Diam}(\mathcal{G})\rho_{\infty}}{2},
	\end{equation}
	where $\mathrm{Diam}(\mathcal{G})$ is called the diameter of $\cal G$.
\end{theorem}
\begin{proof}
	Dynamics of \eqref{powerea} and \eqref{powereb} in a matrix form can be written as
	\begin{equation}
		\label{dP}
		\dot{\hat{P}}_\mathrm{a}=\dot{D}-\mathrm{k}_\mathrm{r}(\hat{P}_\mathrm{a}-D)-\mathrm{k}\mathrm{B}R^{-1}J_T^PT^P,
	\end{equation}
	where $D=\mathrm{col}(D_1,\cdots,D_\mathrm{n})$, $\mathrm{B}$ is a matrix associated with
	the underlying communication graph $\cal G$.
	\par By differentiating $\xi$, one can get
	\begin{equation}
		\begin{aligned}
			\dot{\xi}=f(t,\xi^P)
			=&R^{-1}(t)(\dot{\delta}^P-R\xi^P)\\
			=&R^{-1}(t)(\mathrm{B}^\mathrm{T}(\dot{D}-\mathrm{k}_\mathrm{r}(\hat{P}_\mathrm{a}-D)-\mathrm{k}\mathrm{B}R^{-1}J_{T}^{P}T^{P})
			-R(t)\xi^{P})
		\end{aligned}
	\end{equation}
	where $f(t,\xi^P)$ is piecewise continuous and naturally locally integrable on $t$, and locally Lipschitz on $\xi^P$. Define an open and nonempty set as
	$$\Omega_\xi=\underbrace{(-1,1)\times(-1,1)\cdots\times(-1,1)}_{\mathrm{m}-\mathrm{times}}.$$
	In view of $\rho(0) > |e(0)|$, $\xi_{l}^{P}(0)$ is located in $\Omega_{\xi}$ for
	$l\in\{1, \cdots, \mathrm{m}\}$. Thus, according to Lemma \ref{lem2}, there exists a maximal solution $\xi^P\in\Omega_P$.
	\par Next, select the following candidate Lyapunov function,
	$$V=\frac12(T^P)^\mathrm{T}T^P.$$
	By differentiating $V$ with respect to time, one can obtain that
	$$\begin{aligned}
		\dot{V}=& (T^{P})^\mathrm{T}J_{T}^{P}R^{-1}(t)(\mathrm{B}^\mathrm{T}(\dot{D}-\mathrm{k}_\mathrm{r}(\hat{P}_\mathrm{a}-D) 
		-\mathrm{k}\mathrm{B}R^{-1}(t)J_{T}^{P}T^{P})-\dot{P}(t)\xi_{\delta}) \\
		\text{=}& - \mathrm{k}_\mathrm{r}(T^{P})^\mathrm{T} J_{T}^{P}R^{-1}(t)\mathrm{B}^\mathrm{T}\hat{P}_\mathrm{a}-(T^{P})^\mathrm{T} J_{T}^{P}R^{-1}(t)\dot{R}(t)\xi^{P}\\
		&- \mathrm{k}(T^{P})^\mathrm{T}J_{T}^{P}R^{-1}(t)\mathrm{B}^\mathrm{T}\mathrm{B}R^{-1}(t)J_{T}^{P}T^{P}
		+(T^P)^\mathrm{T}J_TR^{-1}(t)\mathrm{B}^\mathrm{T}(\mathrm{k}_\mathrm{r}D+\dot{D}).
	\end{aligned}$$
	\par Based on the analysis in Remark \ref{rem1}, one can get
	$$\begin{aligned}
		(T^{P})^\mathrm{T}J_{T}R^{-1}(t)\mathrm{B}^\mathrm{T}(\mathrm{k}_\mathrm{r}D+\dot{D})
		\leq\sup_{t\geq0}(\|\mathrm{k}_\mathrm{r}D+\dot{D}\|)\|(T^{P})^\mathrm{T}J_{T}^{P}R^{-1}(t)\mathrm{B}^\mathrm{T}\|,
	\end{aligned}$$
	which is transformed into the following result, according to the Young's inequality,
	$$\begin{aligned}
		&\sup_{t\geq0}(\|\mathrm{k}_\mathrm{r}D+\dot{D}\|)\|(T^{P})^\mathrm{T}J_{T}^{P}R^{-1}(t)\mathrm{B}^\mathrm{T}\|\\
		\leq&\frac{\sup_{t\geq0}(\|\mathrm{k}_\mathrm{r}D+\dot{D}\|)^2}{4\mathrm{k}}
		+\mathrm{k}(T^P)^\mathrm{T}J_T^PR^{-1}(t)\mathrm{B}^\mathrm{T}\mathrm{B}R^{-1}(t)J_T^PT^P.\end{aligned}$$
	\par Besides, the performance function \eqref{pf} has the following property obviously, i.e., $0 < -\frac{\dot{\rho}_{l}(t)}{\rho_{l}(t)} < \lambda$. Thus, the result in below is valid,
	$$\dot{V}\leq-(\mathrm{k}_\mathrm{r}-\lambda)(T^P)^\mathrm{T}J_T^P\xi^P+\frac{\sup_{t\geq0}(\|\mathrm{k}_\mathrm{r}D+\dot{D}\|)^2}{4\mathrm{k}}.$$
	Based on Lemma \ref{lem1} and considering $\mathrm{k}_\mathrm{r}>\lambda$, there exists a
	positive constant $\overline{\epsilon}$ such that $\dot{V}\leq0$ for all $\|T^{P}\|>\bar{\epsilon}$. Given this, there is a time point $t_\mathrm{max}$ such that
	$$|T_l(\xi^P)|\leq\epsilon^*=\max\{\|\xi^P(0)\|,\bar{\epsilon}\},\forall l\in\{1,\cdots,\mathrm{m}\},$$
	within $t\in[0, t_\mathrm{max})$. By taking the inverse of the error transformation function \eqref{pf1}, one can obtain
	$$|\xi_l^P(t)|\leq\xi^*=\tanh(\epsilon^*)<1,$$
	where $\tanh(\epsilon^*)$ denotes the hyperbolic tangent function.
	\par Next, the proof by contradiction can be used to prove that $t_\mathrm{max}$ can be extended to $+\infty$. For $\forall t\in[0, t_\mathrm{max})$, there exists a nonempty and compact subset of $\Omega_{\xi}$ such that $\xi^P(t) \in \Omega_{\xi}^{'}$ and $\Omega_{\xi}^{'}\subset\Omega_{\xi}$, which does not conform to Lemma \ref{lem3}. That is
	to say,
	$$|\xi_l^P(t)|<1,\forall t\geq0,l=1,\cdots,\mathrm{m}.$$
	Thus, it can be concluded that $|\delta_l^P|<\rho_l(t)$. In this way, the steady-state and transient performance of $\delta_l^P$ are regulated by the performance function $\rho_l(t)$.
	\par Finally, explain the performance of steady-state error. Define a new variable $e^{P}=\frac{\bar{\mathbf{1}}^{T}}{\mathrm{n}}(\hat{P}_\mathrm{a}-D)$. Thus, \eqref{dP} can be derived as
	$$\dot{e}^P=-\mathrm{k}_\mathrm{r}e^P,$$
	which implies $e^P$ converges to 0 with rate $\mathrm{k}_\mathrm{r}$. Under a connected graph $\cal G$, the following conclusion holds that consensus error converges to 0 at an exponential rate not less than $\lambda$, i.e.,
	$$\lim\limits_{t\to+\infty}|\hat{P}_{\mathrm{a},i}-\frac{\mathbf{1}^\mathrm{T}}{\mathrm{n}}\hat{P}_{\mathrm{a}}|\leq\frac{\mathrm{Diam}(\mathcal{G})\rho_{\infty}}{2}, i=1,\cdots,\mathrm{n}.$$
	Hence, by selecting $\mathrm{k}_\mathrm{r}>\lambda$, steady state value of average consensus error is bounded by
	$$\lim_{t\to+\infty}|\hat{P}_{\mathrm{a},i}-\frac{\mathbf{1}^\mathrm{T}}{\mathrm{n}}D|\leq\frac{\mathrm{Diam}(\mathcal{G})\rho_\infty}{2}, i=1,\cdots,\mathrm{n}.$$
	So far, the proof is completed.
\end{proof}
\par Based on the results of Theorem \ref{TH2}, a theorem on the SoC balance state estimator is given.
\begin{theorem}\label{TH2}
	Consider a resistive network containing $\mathrm{n}$ BESSs, and SoC and power of each BESS is subjected to \eqref{soc1} and \eqref{power0} respectively. Under the control strategy \eqref{powerec} and \eqref{powered} and Theorem \ref{TH1}, the average SoC can be estimated by agent $i$ for all $i\in\{1,\cdots, \mathrm{n}\}$ with steady-state error as follows,
	\begin{equation}
		\lim\limits_{t\to+\infty}|\hat{E}_{\mathrm{a},i}-E_\mathrm{a}|\leq\frac{\mathrm{Diam}(G)\rho_\infty}{2}.
	\end{equation}
\end{theorem}
\par Due to the formal similarity between the SoC balance estimator and the power sharing estimator, the proof of Theorem \ref{TH2} is similar to Theorem \ref{TH1} and not further elaborated.
\begin{remark}
	\label{re2}
	Obviously, \eqref{powerea}-\eqref{powered} are fully distributed, which is because that only the information of neighbors of each agent is utilized to construct the estimators. Further, the average tracking and consensus performance of this scheme, subjected to different parameters, are decoupled. One performance will not deteriorate when adjusting the other. Specifically, a faster consensus rate is achieved by increasing $\lambda$ and $\mathrm{k}$, while more accurate consensus can be obtained by decreasing $\rho_\infty$. Therefore, there is no compromise between the average tracking and consensus performance when selecting parameters. Besides, from (10) and (13), it can be seen that estimation error can be constrained by $\rho_\infty$ under a given communication topology.
\end{remark}
\begin{remark}
	A finite time average state estimation scheme for BESSs is developed in \cite{Meng2021}, which utilizes the self estimated values and intermediate states communicated by neighbors. More advanced in this paper, the data packets transmitted over the communication network only contain the average estimated states, thereby reducing communication overhead. Besides, the harmful phenomenon of chattering can avoid under our scheme. As well, prior knowledge of the derivative of the load is not required.
\end{remark}
\section{Some Simulation Cases}\label{5}
\par In this section, several cases are designed at the hourly scale, i.e., $\mathrm{T}=1$, to verify the proposed estimators. A resistance network containing 4 BESSs is selected as the load distribution exhibited in Fig. \ref{comm}, where the communication topology of agents governing BESSs is illustrated.
\par Next, four cases will be studied to illustrate the effectiveness, progressiveness, scalability and robustness of the designed estimators.
\par The effectiveness of the two designed estimators are verified in Cases 1.
\par Compared with \cite{9925608}, the progressiveness of the two designed estimators is illustrated in Case 2.
\par The robustness of the two designed estimators are explained in Case 3. The simulation is initially executed according to Case 1, followed by simulating failure of certain node at a certain moment to verify robustness.
\par The scalability of the two designed estimators is explained in Case 4. 8 BESSs are newly added to be simulated.
\begin{figure}
	\centering
	\includegraphics[width=8cm]{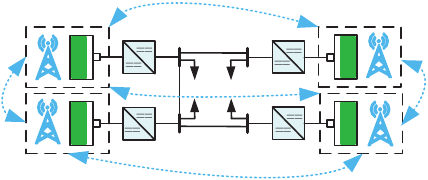}
	\caption{Four BESSs in a resistance network and their communication graph
		topology.}
	\label{comm}
\end{figure}
\begin{figure}
	\centering
	\includegraphics[width=8cm]{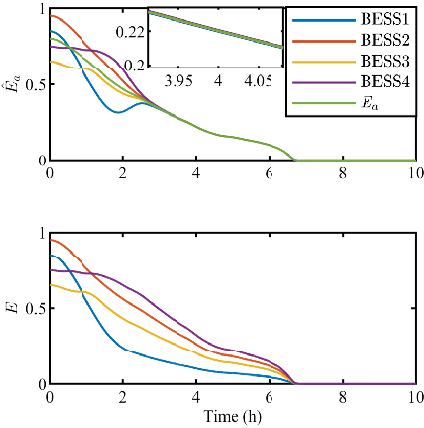}
	\caption{Evolution of SoC balance estimation and SoC of each BESS in Case 1.}
	\label{case1a}
\end{figure}
\begin{figure}
	\centering
	\includegraphics[width=8cm]{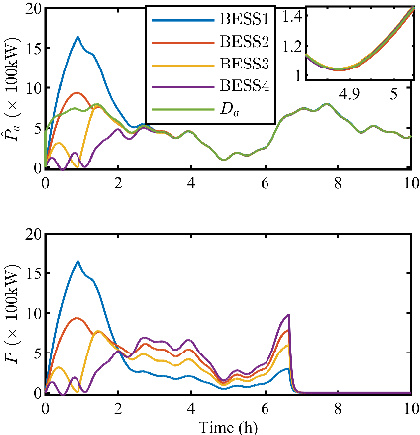}
	\caption{Evolution of power sharing estimation and power of each BESS in Case 1.}
	\label{case1b}
\end{figure}
\begin{figure}
	\centering
	\includegraphics[width=8cm]{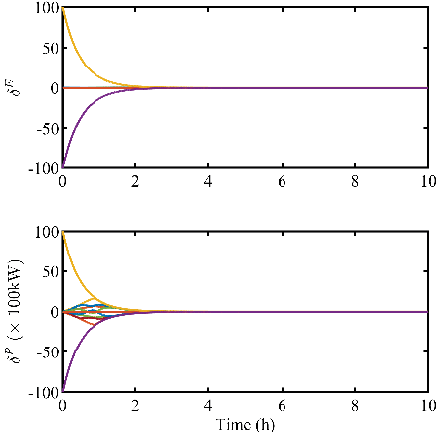}
	\caption{Evolution of consensus error in Case 1.}
	\label{case1d}
\end{figure}
\begin{figure}
	\centering
	\includegraphics[width=8cm]{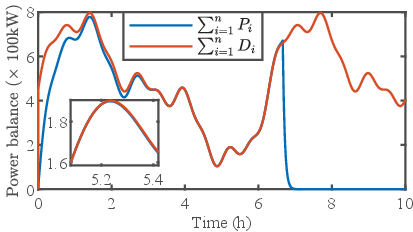}
	\caption{Power demand-supply balance in Case 1.}
	\label{case1c}
\end{figure}
\subsection{Case 1. The Test of Effectiveness}
\par Driven by a load with bounded first derivative, a simulation is executed under the designed estimators with continuous dynamics \eqref{powerea}-\eqref{powered}. The simulation results are shown in Figs. \ref{case1a}-\ref{case1c}.
\par From Figs. \ref{case1a} and \ref{case1b}, it can be concluded that power sharing state and SoC balance state can be well estimated by each agent with almost zero error. Meanwhile, SoC of each battery almost simultaneously decreases to the threshold and remains unchanged thereafter. This indicates that no BESS will exit early due to insufficient power. In addition, consensus errors, as shown in Fig. \ref{case1d}, are limited to evolve within predetermined regions, and supply and demand balance, as shown in Fig. \ref{case1c}, is well ensured.
\subsection{Case 2. Comparison with other methods}
\par To illustrate the progressiveness, we will compare the effect of the estimators designed in \cite{9925608}. The simulation results of this estimator under the same gains as those in Case 1 are shown in Figs. \ref{case2a}-\ref{case2c}.
\par Comparing the results in Case 1 and Case 2, it can be seen that the estimator designed in this paper can more accurately estimate the power sharing state and SoC balance state in real-time. From the simulation results, it can be seen that under the driving force of the scheme designed in \cite{9925608}, the estimated values of SoC balance state and average power exhibit varying degrees of jitter, which leads to jitter in the output power of each battery. In addition, it can be clearly compared from the local magnified images in the simulation results that the scheme designed in this article has high control accuracy.
\par Furthermore, according to the results of \cite{9925608}, its steady-state performance is affected by gains. However, according to Theorems \ref{TH1} and \ref{TH2} and Remark \ref{re2} in this paper, the steady-state performance of each designed estimator is preset and decoupled from the gain.
\begin{figure}
	\centering
	\includegraphics[width=8cm]{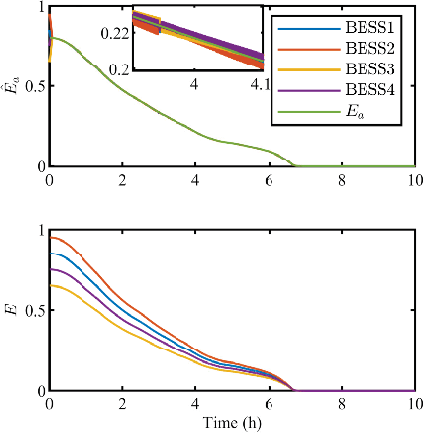}
	\caption{Evolution of SoC balance estimation and SoC of each BESS in Case 2.}
	\label{case2a}
\end{figure}
\begin{figure}
	\centering
	\includegraphics[width=8cm]{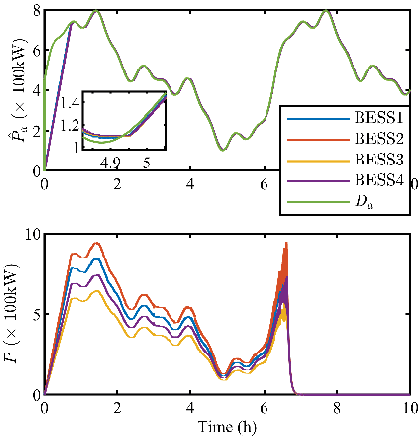}
	\caption{Evolution of power sharing estimation and power of each BESS in Case 2.}
	\label{case2b}
\end{figure}
\begin{figure}
	\centering
	\includegraphics[width=8cm]{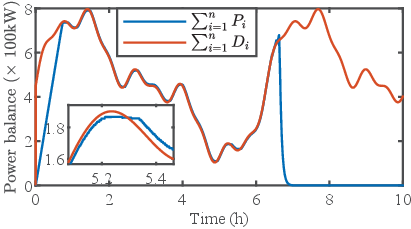}
	\caption{Power demand-supply balance in Case 2.}
	\label{case2c}
\end{figure}
\subsection{Case 3. The Test of Robustness}
\par To test robustness, the simulation under the communication failure of the agent managing BESS 3 at $t = 5h$ is conducted. In this case, resimulate based on Case 1 the estimators adopted in this paper. The simulation results are shown in Figs. \ref{case3a} and \ref{case3b}. 
\begin{figure}
	\centering
	\includegraphics[width=8cm]{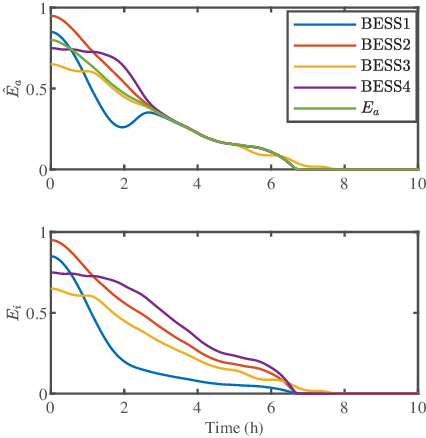}
	\caption{Evolution of SoC balance estimation and SoC of each BESS in Case 3.}
	\label{case3a}
\end{figure}
\begin{figure}
	\centering
	\includegraphics[width=8cm]{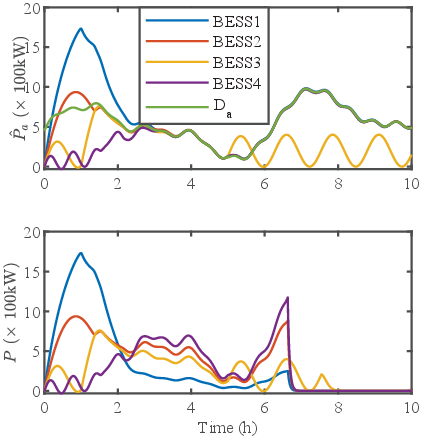}
	\caption{Evolution of power sharing estimation and power of each BESS in Case 3.}
	\label{case3b}
\end{figure}
\par From the simulation results, it can be seen that although the agent managing BESS 3 fails, resulting in it no longer providing power to the resistive network, it has no effect on the estimation of the power sharing state and SoC balance state by other agents. That is to say, while the gain remains unchanged and the estimator does not need to be initialized, the estimation of the power sharing state and SoC balance state by each agent is still effective. In this case, BESSs that are still online will eventually be offline at the same time. That is, the designed estimators robust to agent failure.
\begin{figure}
	\centering
	\includegraphics[width=8cm]{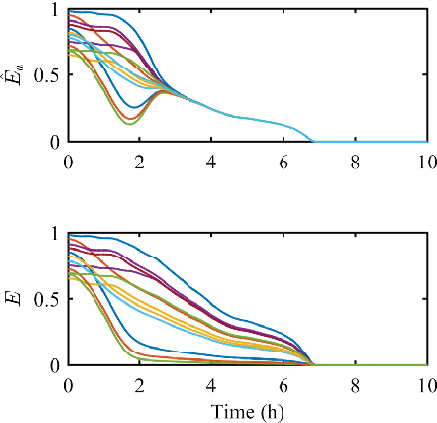}
	\caption{Evolution of SoC balance estimation and SoC of each BESS in Case 4.}
	\label{case4a}
\end{figure}
\begin{figure}
	\centering
	\includegraphics[width=8cm]{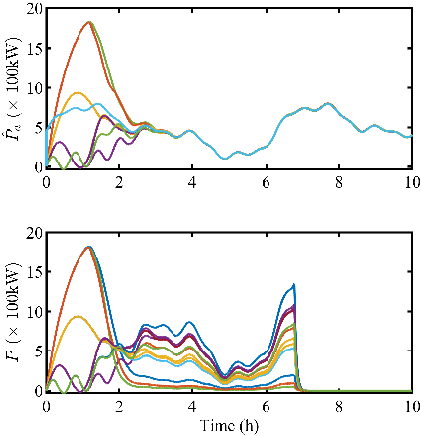}
	\caption{Evolution of power sharing estimation and power of each BESS in Case 4.}
	\label{case4b}
\end{figure}
\begin{figure}
	\centering
	\includegraphics[width=8cm]{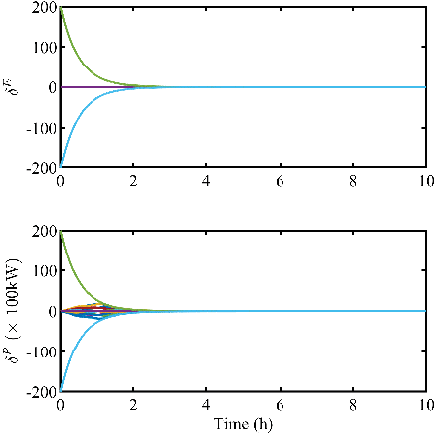}
	\caption{Evolution of consensus error in Case 4.}
	\label{case4d}
\end{figure}
\begin{figure}
	\centering
	\includegraphics[width=8cm]{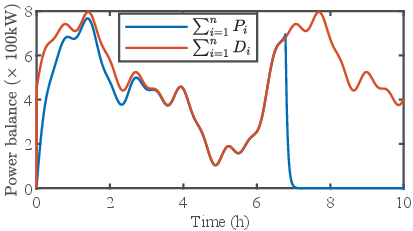}
	\caption{Power demand-supply balance in Case 4.}
	\label{case4c}
\end{figure} 
\subsection{Case 4. Wide Area Microgrid Testing: Scalability}
\par To meet the needs of the wide area microgrid, two sets of BESSs with the same parameters will be added to Case 1, so that the entire resistance network will have 12 BESSs. Correspondingly, the number of agents has also increased to 12. Under the same parameters as Case 1, the simulation results are shown in Figs. \ref{case4a} and \ref{case4b}.
\par From the simulation results, it can be seen that under the driving force of the designed scheme, the estimation of the SoC balance state and average power of this wide area system can be well completed, and the consensus error can be limited within the preset area. SoCs almost decrease to 0 at the same time, after which the output power of BESSs will be 0. The balance of power supply and demand can be ensured with sufficient energy storage. Thus, it can be seen that the designed estimators can meet the needs of the wide area microgrid. This is entirely due to the scalability of distributed algorithms for MASs.
\section{Conclusions}\label{6}
\par In this paper, two distributed estimators with prescribed dynamic and steady-state performance for BESSs with an unified relative discharge rate are designed to ensure that each agent managing BESS can obtain SoC balance and load sharing states. In this way, the consensus and tracking performance of the two estimators are decoupled, thus being constrained by different parameters. However, the continuous communication mechanism used in this article will increase the communication burden. In addition, various forms of cyber attacks that occur from time to time also pose significant risks to communication security. Therefore, how to achieve a solution with preset performance, intermittent communication, and resilience should be a promising research issue. 

\backmatter

\begin{center}
	N\footnotesize{OMENCLATURE}
	\normalsize
	\begin{tabbing}
		\hspace{2cm} \= \kill
		BESS\> battery energy storage system\\
		SoC\>  the State-of-Charge\\
		PPC\> prescribed performance control\\
		MAS\> multi-agent system\\
		power\> proportional output power\\
		$E_i$, $I_i$\> SoC and output current\\
		$\mathrm{Q}_i$, $\mathrm{V}_i$\> the capacity and terminal voltage\\
		$\tilde P_i$, $P_i$\> output power and its proportional value\\
		$\mathrm{K}_i^E$\> a constant defined as $\mathrm{K}_i^E=\frac{1}{\mathrm{Q}_i\mathrm{V}_i}$\\
		$D_i$, $P_a$, $E_a$\> the proportional load, the average power and SoC\\
		$\hat P_{a,i}$, $\hat E_{a,i}$\> the estimated average power and SoC \\
		$\epsilon_E$, $\epsilon_P$\> the pre-specified values\\
		$\cal G$, $\cal V$, $\cal E$\> communication graph, vertex set and edge set\\
		$\cal A$, ${\cal N}_i$\> the adjacency matrix of $\cal G$ and the neighbor set\\
		$\rho_{ij}$, $\rho_l$, $R$\>the performance function and its matrix of power sharing error\\
		$\rho_0$, $\rho_{\infty}$, $\lambda$\> the initial, final values and decay rate of $\rho(t)$\\
		$\xi(t)$, $\xi^P$, $\xi^E$\> a modulated error and its vector forms\\
		$T(\xi)$, $T^P$, $T^E$\> a bijective mapping and the vector forms\\
		$J_T$, $J_T^P$, $J_T^E$\> the Jacobian derivative of $T(\xi)$, $T^P$, $T^E$\\
		$\bar{\epsilon}$, $\mathrm{K}$, $\mathrm{F}$\> some constant defined in Lemma \ref{lem1}\\
		$\epsilon^*$, $\xi^*$ \> some constant defined in the proof\\
		$\mathrm{a}_1$, $\mathrm{a}_2$\> maximum and minimum allowed SoC values\\
		$z_i^P$, $z_i^E$\> the intermediate states of the two estimators\\
		$\delta_l^E$, $\delta_l^P$\> consensus errors for estimation\\
		$\delta^E$, $\delta^P$\> consensus error vectors for estimation\\
		$\Omega_\xi$, $\Omega_\xi^{'}$\> an open and nonempty set and its subset\\
		$V$, $\dot V$\> a Lyapunov function and its derivative\\
		$e^P$\> a tracking error vector\\
		$\mathrm{m}$, $\mathrm{n}$\> number of communication links and of agents\\
		$l$, $i$\> the index of communication links and agents\\
		$\mathrm{k}^P$, $\mathrm{k}^P_\mathrm{r}$\> the gains of the power sharing state estimator\\
		$\mathrm{k}^E$, $\mathrm{k}^E_\mathrm{r}$\> the gains of the SoC balance state estimator\\
	\end{tabbing}
\end{center}
\bmhead{Data availability statement}
 Data sharing is not applicable to this article as no new data were created or analyzed in this study.
\bibliography{re0}

\end{document}